\documentclass[11pt]{article}
\usepackage[margin= 1in]{geometry}
\IfFileExists{lmodern.sty}{\usepackage[T1]{fontenc}\usepackage{lmodern}}{}
\usepackage{amsmath,amssymb,amsthm,mathtools}
\usepackage{microtype}
\usepackage{enumitem}
\usepackage{titlesec}
\usepackage[colorlinks=true,urlcolor=blue,linkcolor=blue]{hyperref}

\titleformat{\section}{\large\bfseries}{\thesection.}{0.5em}{}
\titlespacing*{\section}{0pt}{9pt plus 2pt}{3pt}
\setlist{itemsep=1pt,topsep=2pt,parsep=0pt}

\newtheoremstyle{compact}{5pt}{5pt}{\itshape}{}{\bfseries}{.}{0.5em}{}
\theoremstyle{compact}

\newtheorem{lemma}{Lemma}

\newcommand{\dd}{\mathrm d}
\newcommand{\diag}{\operatorname{diag}}
\newcommand{\tr}{\operatorname{tr}}
\newcommand{\one}{\mathbf 1}
\newcommand{\R}{\mathbb R}
\newcommand{\cL}{\mathcal L}
\newcommand{\cM}{\mathcal M}
\newcommand{\cD}{\mathcal D}
\newcommand{\ip}[2]{\langle #1,#2\rangle}
\newcommand{\T}{^{\top}}

\begin{document}

\title{On the Guo–Fang–Lu Algorithm for Koml\'os Discrepancy}
\date{}
\author{Nikhil Bansal}
\maketitle
\begin{abstract}

We give an exposition of the recent polynomial time algorithm of Guo, Fang, and Lu~\cite{GuoFangLu26} for the
Koml\'os problem. We simplify various arguments, and highlight the key new spectral potential idea, and how the algorithm follows naturally from it.
\end{abstract}
\section{Introduction}
There has been amazing recent progress on the Koml\'os problem. Guo--Fang--Lu \cite{GuoFangLu26exist},
building on the work of Smirnov and Vershynin \cite{SmirnovVershynin26},
showed the existence of $O(1)$ discrepancy coloring. See also \cite{KaringulaLovett26, Bandeira26blog} for simpler expositions. However, it was unclear how to use this approach to find such a coloring efficiently in polynomial time.

More recently, Guo--Fang--Lu \cite{GuoFangLu26} also gave such an algorithm based on the Brownian walk approach in discrepancy. Their work builds on the work of Guillen and Kobzar~\cite{GuillenKobzar26}, which shows $O(1)$ discrepancy when the $\pm 1$ signs are relaxed to be unit complex vectors.
The goal of this note is to make the remarkable new ideas in \cite{GuoFangLu26} more accessible. It was motivated by a reading group on discrepancy at the Simons Institute.

The author used Astra 6.0 to understand the paper \cite{GuoFangLu26}, and to help with several calculations in this note. The exposition is completely by the author, who takes full responsibility for any errors.

\medskip
{\bf The Problem.} In the Koml\'os problem, we are given a matrix $A \in \R^{m\times n}$ with columns of length at most $1$, and the goal is to find a  signing (aka coloring) $x\in \{-1,1\}^n$ with low discrepancy $\|Ax\|_\infty$. For cleaner exposition, 
we will focus on the Beck-Fiala setting, where $A$ has entries in $\{0,1\}$, with at most $k$ non-zero entries in every column, and prove an algorithmic $O(\sqrt{k})$ bound.
We assume throughout that $k\geq k_0$ for some large enough constant $k_0$.
The ideas extend directly to the Koml\'os problem, and we sketch this in Section \ref{sec:komlos}.

\medskip
{\bf High-level Idea.} 
We describe the continuous time process underlying the algorithm, which is cleaner. The polynomial time implementation follows by standard discretization.

The process starts from $x_0={\bf 0}$ at time $t=0$, and evolves the coloring $x_t \in [-1,1]^n$ over time using a suitable Brownian motion.
Once a coordinate reaches $\pm1$, it is frozen and no longer updated. 
As the coloring evolves, call a row {\em dangerous} if its (regularized)  discrepancy gets close to the target $\beta=O(k^{1/2})$, while it still has many uncolored coordinates. The key idea will be to ensure that at any time $t$, there are at most $N_t/4$ dangerous rows, where $N_t$ is the number of alive coordinates. Given this property, one can walk in a suitable subspace, so that the discrepancy of such dangerous rows only decreases.

Previous works also use a similar approach \cite{BansalJiang26, GuillenKobzar26}, but (roughly) the problem is that the number of dangerous rows gets harder to control as $N_t$ gets smaller over time.\footnote{To handle this, they require that $k = \Omega(\text{polylog}(n))$, or they  relax the signs to be unit complex numbers.}
The key new idea here is to get this control by bounding a single spectral quantity: the operator norm $\|W_t\|$ of a Gram matrix $W_t$ formed by the rows  weighted (cleverly) by how dangerous they are currently. Given this, the algorithm is designed cleverly to ensure that $\|W_t\|\leq 1$ holds at all times.

\section{The Walk}\label{sec:walk}
We now describe the continuous process. We begin with some notation and definitions, and then describe the spectral potential and why it is useful.

\smallskip

\textbf{Notation.}
 We index rows by $i$ and columns by $j$. Let $a_i$ denote row $i$ of $A$.
Introducing rows $-a_i$ if necessary,\footnote{The column sparsity is now $2k$ and $A$ has $\{-1,0,1\}$ entries. We still call this matrix $A$.} 
it suffices to upper bound the (one-sided) discrepancy $\ip{a_i}{x}$ over all rows.

For a fractional coloring $x_t \in [-1,1]^n$ at time $t$, let $V_t=\{j:|x_{t}(j)|<1\}$ denote the set of alive coordinates and let $N_t=|V_t|$. 
For a row $i$, let
 $S_{i,t}=\{j\in V_t:a_{ij}\neq0\}$ be its support on alive coordinates and $s_{i,t}=|S_{i,t}|$ denote its {\em size}. Finally, let $E_{i,t}=\sum_{j\in S_{i,t}}(1-x_t(j)^2)$ denote its remaining {\em energy}, and $d_{t}(i)=\ip{a_i}{x_t}$ denote its discrepancy. Note that $E_{i,t}\leq s_{i,t}$.

\medskip

\textbf{Row classes.} At any time $t$, we classify rows based on their size $s_{i,t}$ as:

(i) {\em large} ($i\in\cL_t$) if $s_{i,t} >16k$, 

(ii)  
{\em medium}  ($i\in\cM_t$) if $s_{i,t} \in [20 k^{1/2},16k]$,

(iii) {\em small} if  $s_{i,t} <20 k^{1/2}$.

Note that $s_{i,t}$ can only decrease over time.
It is easy to ensure that any row incurs no discrepancy while it is large. Also, once a row is small, we can  ignore it as it can incur at most $O(k^{1/2})$ additional discrepancy. So the action will be in handling medium rows.

\medskip

{\bf Slack and Dangerous Rows.}
Let $\beta=50 k^{1/2}$.
The algorithm will ensure the following for each medium row $i$ and each time $t$:
\begin{equation}\label{budget}
 d_t(i)+\frac{E_{i,t}}{\sqrt{s_{i,t}}}+5\sqrt{s_{i,t}}\le\beta. 
\end{equation}
As $E_{i,t},s_{i,t}\geq 0$, this will give the desired $O(k^{1/2})$ discrepancy.
To this end, let us define the \emph{slack} \[\sigma_{i,t} :=\frac{\beta-d_t(i)}{\sqrt{s_{i,t}}}-\frac{E_{i,t}}{s_{i,t}}-5.\]
So requirement \eqref{budget} is equivalent to $\sigma_{i,t}\geq 0$.
Note that $\sigma_{i,t}$ can change discontinuously when some variable $x_t(j)$ reaches $\pm 1$ (and $s_{i,t}$ decreases by $1$). We call this a {\em freezing} event. Also, when row $i$ first becomes medium at time $t$ ($t=0$ if it is medium in $A$), as $d_{i,t}=0$ and $s_{i,t} \leq 16k$, and $E_{i,t}\leq s_{i,t}$, we have that  $\sigma_{i,t} \geq 50/4-6\geq 6$, and hence far enough from $0$.

 We call a medium row $i$ \emph{dangerous} at time $t$ if $\sigma_{i,t} \leq 1/20$. Let 
 $\cD_t$ denote the set of dangerous rows at time $t$.
The potential below is the key new ingredient to bound $|\cD_t|$ in terms of $N_t$. 

\medskip 

\textbf{A Spectral Potential.} 
Define the following matrix $W_t \in \R^{V_t \times V_t}$ at time $t$.
\begin{equation}\label{matrix}
 W_t =\frac{\one_{V_t}\one_{V_t}\T}{100n}+\sum_{i\in\cL_t}\frac{\diag(\one_{S_{i,t}})}{100k}+\sum_{i\in\cM_t}W_{i,t},
\end{equation}
where $W_{i,t}=10  e^{-\sigma_{i,t}} \,s_{i,t}^{-2}\,\one_{S_{i,t}}\one_{S_{i,t}}\T$.

The third term in \eqref{matrix} is the important one.
Note that $W_{i,t}$ is a rank one matrix corresponding to medium row $i$, and the term $\exp(-\sigma_{i,t})$ gets larger  as the slack gets smaller.

Also, each entry in $W_{i,t}$ is non-negative.
The first term in \eqref{matrix} is negligible and is there to make each entry of $W_t$ (strictly) positive. By the Perron--Frobenius theorem, the top eigenvalue  $\lambda_t = \lambda_{\max}(W_t ) = \|W_t\|$ is simple, and depends smoothly on $x$. 
\allowdisplaybreaks
It is also easily verified\footnote{For any medium row one has $W_{i,0} \preceq \diag(\one_{S_i})/(100k)$, so
 $W_0\preceq \one\one\T/(100n)+\sum_{i\in\cL\cup\cM}\diag(\one_{S_{i,0}})/(100k)
 \preceq (1/100+ 2k/(100k))I=3/100\,I$,
where the second step uses that each coordinate lies in $\leq 2k$ rows.}
that $\lambda_0\leq 3/100$ initially, at $t=0$. 

\medskip 
{\bf Why is $W_t$ useful.}
The algorithm will maintain that $\lambda_t \le 1$ for all times $t$.

This spectral control has several nice properties. A general difficulty in discrepancy algorithms is that one lacks much control on which variables reach $\pm 1$ over time. However, $\lambda_t\leq 1$ ensures that
irrespective of which $N_t$ columns are alive at time $t$, at most $N_t/8$ rows can be dangerous. 
Indeed, as  $\ip{\one_{V_t}}{W_{i,t}\one_{V_t}}=10e^{-\sigma_{i,t}}$ for each medium row $i$, one has
\begin{equation}\label{protected}
 \sum_{i\in\cM_t}10e^{-\sigma_{i,t}}\le\ip{\one_{V_t}}{W_t \one_{V_t}}\le \lambda_t \|\one_{V_t}\|^2=   \lambda_t N_t\le N_t.
\end{equation}
As $10e^{-\sigma_{i,t}} \geq 10e^{-1/20} \geq 8$ for each dangerous row $i \in \cD_t$, this implies $|\cD_t| \leq N_t/8$. 

Previous approaches bounded $|\cD_t|$ indirectly by requiring that {\em every} column $j$ have $\ll k$ of its entries  lie in dangerous rows.
However, the condition $\lambda_t\leq 1$ is more relaxed and does not require such strong per-column control.
Finally, $\lambda_t$ is also well-behaved analytically.
\subsection{The Walk} 
This motivates the following natural process.

a. Start with $x_0={\bf 0}$ at $t=0$. 

b. While $N_t \geq 20k^{1/2}$, set $\dd x_t=h_t \,\dd B_t$, where $B_t$ is a standard 1-d Brownian motion and $h_t\in\R^{V_t}$ is any unit vector satisfying:
\begin{enumerate}
\item $h_t \perp x_{V_t}$; \hfill (progress)
\item $h_t\perp a_{i,{V_t}}$ for every $i\in\cL_t$; \hfill (protect large rows)
\item $h_t\perp\nabla\sigma_{i,t}$ for every $i\in\cD_t$; \hfill (protect dangerous rows)
\item if $\lambda_t \geq 1$, then $\partial_{h_t}\lambda_t=0$ and $\partial_{h_t}^2\lambda\le0$. \hfill ($\|W_t\|$ does not grow)
\end{enumerate}
\,\,\,\,\, When a coordinate reaches $\pm1$, it leaves $V_t$, rows are reclassified, and $s_{i,t},W_t$ are updated. 

c. When no large or medium rows remain, set the alive variables arbitrarily to $\pm 1$.

\section{Analysis}\label{sec:invariants}
We now show that $h_t$ always exists, and that the walk above produces an $O(k^{1/2})$ discrepancy coloring.
This will follow from the next two Lemmas.
\begin{lemma}\label{lem:invariants}
For any time $T$, if some $h_t$ satisfying Conditions 1--4 exists at every $t<T$, then $\sigma_{i,t}\ge 0$ and $\lambda_t\le 1$,
for every medium row $i$ and time $t \leq T$, 
\end{lemma}
\begin{lemma}
\label{lem:counting}
At any time $t$ at which the algorithm is still running (so that $N_t\ge 20k^{1/2}$), if $\lambda_t\le 1$, then there is some unit vector $h_t\in\R^{V_t}$ satisfying Conditions 1--4. 
\end{lemma}

Indeed, together these imply that $h_t$ always exists. This gives the claimed discrepancy bound as follows. 
By Condition 1 for $h_t$ $\dd\|x\|^2=2\ip{x}{h}\,\dd B_t+\|h\|^2\dd t=\dd t$, so the algorithm terminates by time $n$.
By Condition 2 for $h_t$, a row incurs no discrepancy while it is large, and Lemma \ref{lem:invariants} ensures that $\sigma_{i,t}\geq 0$ while it is medium.

\begin{proof} (Lemma \ref{lem:invariants})
We first consider how the various (discontinuous) structural changes due to some coordinate freezing can affect
$\sigma_{i,t}$ and $\lambda_t$.

 (i) Medium row $i$ becomes small: This only reduces $\lambda_t$ as the psd term $W_{i,t}$ disappears from $W_t$.

 (ii) Large row $i$ changes to medium: 
When this happens $\sigma_{i,t}\geq 50/4 -6 \geq 6$. Moreover, $\lambda_t$ can only decrease as $s_{i,t}=16k$ and thus 
\[ W_{i,t} = 10e^{-\sigma_{i,t}} s_{i,t}^{-2} \one_{S_{i,t}}\one_{S_{i,t}}\T \preceq 10 e^{-6}s_{i,t}^{-1} \diag(\one_{S_{i,t}})  \preceq \frac{1}{100k} \diag(\one_{S_{i,t}}).\] 

(iii) Variable $j$ freezes: 
Here $d_{i,t}$ and $E_{i,t}$ are unaffected, and $s_{i,t}$ decreases by $1$ if row $i$ contains $j$.
We show that $\sigma_{i,t}$ can only increase
 and $\lambda_t$ can only decrease.
 Indeed (dropping $i,t$ to avoid clutter), and using that $\sigma>0$ and $E/s\leq 1$, we have 
\[
 \frac{\dd \sigma}{ds}   =-\frac{\beta-d}{2s^{3/2}}+\frac{E}{s^2} =-\frac{\sigma+5-E/s}{2s} < 0,\]
As $ds<0$ for us, this gives that $\sigma$ can only increase. Similarly, 
 \[\frac{d}{ds}\log (10e^{-\sigma}s^{-2}) =-\frac{\dd \sigma}{ds}-\frac2s=\frac{\sigma+1-E/s}{2s}\ \ge\ 0.
\]
So $W_t$ can only decrease in the psd ordering. Finally,
since $V$ changes to $V\setminus \{j\}$, by Cauchy interlacing, restricting $W_t$ further to the principal submatrix $V'\times V'$ can only decrease $\lambda_t$ further.

\medskip

{\bf Continuous motion.} 
We now consider the effect of the motion $dx_t = h_t dB_t$. Again we drop $t$ to avoid clutter.
For a row $i$ we have the directional derivatives:

(i) $\partial_hd_i=\ip{a_{i,V}}{h}$ and  $\partial_h^2d_i=0$,

(ii) $\partial_hE_i=-2\ip{x_{S_i}}{h_{S_i}}$ and $\partial_h^2E_i=-2\|h_{S_i}\|^2$.

Note that $s_{i,t}$ stays fixed.
So by It\^o's formula, 
\begin{align*}
 \dd\sigma_i=\ip{\nabla\sigma_i}{h}\,\dd B_t+ \|h_{S_i}\|^2/s_i\,\dd t, \quad \text{and} \quad 
 \dd\lambda=\partial_h\lambda\,\dd B_t+\tfrac12\partial_h^2\lambda\,\dd t.
\end{align*}
When $\sigma_i \leq 1/20$, Condition 3 kills  the $\ip{\nabla\sigma_i}{h}\dd B_t$ term, so $\sigma_i$ can only increase.
Similarly, when $\lambda\geq 1$, Condition 4  ensures (deterministically) that $\lambda$ cannot increase.  
\end{proof}

\subsection{Proof of Lemma \ref{lem:counting}}
We now show the existence of $h_t$ satisfying Conditions 1--4, using a dimension counting argument.
We consider a fixed $t$, and drop the subscript $t$ throughout to avoid clutter.

Conditions 1--3 are easy and they impose at most $1+N/16+N/8 < N/3$ (as $N$ is large enough) linear constraints. Indeed,
Condition 1 is a single constraint. Condition 2 for large rows imposes at most $N/16$ constraints ($a_i$ and $-a_i$ give the same constraint). For Condition 3, by \eqref{protected} and as $\lambda\leq 1$ there are at most $N/8$ dangerous rows.

Thus if $\lambda <1$, Condition 4 is vacuous, and we are already done. So, henceforth we assume that $\lambda \geq 1$. 
Even though Condition 4 is not linear in $h$, it is handled by the following lemma.

\begin{lemma}\label{lem:spectral}
Let $k$ be sufficiently large constant. If $\lambda \geq 1$, there is a subspace $H\subseteq\R^V$ with $\dim H>N/3$ such that every $h\in H$ satisfies $\partial_h\lambda=0$ and $\partial_h^2\lambda\le0$.

In particular, as Conditions 1--3 impose less than $N/3$ linear constraints on $H$, this implies Lemma~\ref{lem:counting}. 
\end{lemma}

\begin{proof}
Let $\lambda=\lambda_1>\lambda_2\ge\cdots\ge\lambda_N\ge0$ be the eigenvalues of $W$, with orthonormal eigenvectors $v=v_1,v_2,\ldots,v_N$.
We need to understand how $h$ affects $\lambda$. Let $\partial_h$ denote the derivative along $h$.

Recall the perturbation formula for a simple top eigenvalue,\footnote{Differentiating $Wv=\lambda v$ and $v\T v=1$ gives $\partial_h\lambda=v\T(\partial_hW)v$ and $\partial_hv=R\,(\partial_hW)v$. Differentiating $\partial_h\lambda=v\T(\partial_hW)v$ once more gives the second formula.}
\begin{equation}\label{perturb}
\partial_h\lambda=v\T(\partial_hW)v,\qquad \partial_h^2\lambda=v\T(\partial_h^2W)v+2\,\bigl((\partial_hW)v\bigr)\T R\,\bigl((\partial_hW)v\bigr),
\end{equation}
where $R=\sum_{\ell\ge2} v_\ell v_\ell\T/ (\lambda-\lambda_\ell)$. 

 As $W_i$ is proportional to $e^{-\sigma_i}$, and using 
$\partial_h^2\sigma_i=2\|h_{S_i}\|^2/s_i$, we have
\[
 \partial_hW=-\sum_{i\in\cM}\ip{\nabla\sigma_i}{h}\,W_i,\qquad
 \partial_h^2W=\sum_{i\in\cM}\bigl(\ip{\nabla\sigma_i}{h}^2-2\|h_{S_i}\|^2/s_i\bigr)W_i .
\]
Notice that $\partial_h\lambda=0$ only imposes a single linear constraint. Thus the main task is to control $\partial_h^2\lambda$.

\medskip

{\bf Notation.} It is useful to introduce some notation to simplify the expression. For a medium row $i$, write $W_i=c_ic_i\T$ with $c_i=(10e^{-\sigma_i})^{1/2}s_i^{-1}\one_{S_i}$. Let $C$ be the matrix with columns $c_i$, $i\in\cM$, so that $CC\T=\sum_{i\in\cM}W_i\preceq W$. 
For a direction $h$, let $z\in\R^{\cM}$ be the vector with entries $z_i=-\ip{\nabla\sigma_i}{h}\ip{c_i}{v}$, and let $D$ be the diagonal matrix
\[
 D=\sum_{i\in\cM} 2\,v\T W_iv\,\diag(\one_{S_i}/s_i).
\]
Then $(\partial_hW)v=Cz$. As $\ip{c_i}{v}^2=v\T W_iv$, using the definition of $z$ and $D$, we can write $v\T(\partial_h^2W)v=\|z\|^2-h\T Dh$.
So by \eqref{perturb}, 
\begin{equation}\label{second}
 \partial_h^2\lambda \;=\;\Bigl(\,
 \underbrace{\|z\|^2}_{T_1}
 \;-\;\underbrace{h\T Dh}_{T_2}
 \;+\;\underbrace{2\,(Cz)\T R\,(Cz)}_{T_3}\Bigr).
\end{equation}
{\bf The subspace $H$.} Let $F=\operatorname{span}\{v_\ell:\lambda_\ell\ge\lambda/10\}$ be the top of the spectrum of $W$; it contains $v$. Write $\|z\|^2=h\T Uh$, where $U=\sum_{i\in\cM}(v\T W_iv)\,\nabla\sigma_i\nabla\sigma_i\T$. Let $J=\{j:D_{jj}>0\}$. If $D_{jj}=0$, then $v\T W_iv=0$ for every $i$ with $j\in S_i$, so the $j$th row and column of $U$ vanish as well. Let $M=D_J^{-1/2}U_JD_J^{-1/2}\succeq0$, and define
\begin{align*}
 H_1&=\{h:\ \Pi_FCz=0\},\\
 H_2&=\bigl\{h:\ D_J^{1/2}h_J  \text{ is orthogonal to every eigenvector of $M$ with eigenvalue}>9/11\bigr\}.
\end{align*}
Thus $H_1$ blocks the top of the spectrum of $W$, killing the part of $T_3$ with small $\lambda-\lambda_\ell$. We show:
\begin{enumerate}[label=(\alph*),leftmargin=2em]
\item on $H_1$: the term $T_3 \leq (2/9)\|z\|^2$;
\item on $H_2$: the term $\|z\|^2\le (9/11)\,h\T Dh$;
\item $H_1$ has codimension at most $N/100$;
\item $H_2$ has codimension at most $5N/8$.
\end{enumerate}

Then $H=H_1\cap H_2$ gives the desired subspace, as $\dim(H) \geq N-N/100-5N/8> N/3$, and by \eqref{second} every $h\in H$ has
$ \partial_h^2\lambda \leq  (11/9)\|z\|^2-h\T Dh\ \le\ 0.$

\medskip
{\bf Proof of (a).} Split $R=R_F+Q$, where
\[
 R_F=\sum_{\ell\ge2:\ v_\ell\in F}v_\ell v_\ell\T/(\lambda-\lambda_\ell),\qquad
 Q=\sum_{\ell:\ \lambda_\ell<\lambda/10}v_\ell v_\ell\T/(\lambda-\lambda_\ell).
\]
On $H_1$, as $R_F\,Cz=0$, the term $T_3 = 2(Cz)\T Q(Cz)$. As $CC\T\preceq W$ and $Q$ commutes with $W$,
\[
 (Cz)\T Q\,(Cz)\le\|Q^{1/2}WQ^{1/2}\|\,\|z\|^2 =\max_{\ell:\,\lambda_\ell<\lambda/10}\frac{\lambda_\ell}{\lambda-\lambda_\ell}\,\|z\|^2\ \le\ \frac19\,\|z\|^2.\]
{\bf Proof of (b).} For $h\in H_2$, the vector $f=D_J^{1/2}h_J$ lies in the span of the eigenvectors of $M$ with eigenvalue at most $9/11$. So $\|z\|^2=h\T Uh=f\T Mf\le\frac9{11}\|f\|^2=\frac9{11}\,h\T Dh$.

\medskip
{\bf Proof of (c).} Let $r$ denote the number of eigenvalues of $W$ with $\lambda_\ell\ge\lambda/10$. We show that $r \leq N/100$. Split $W=W_{\rm med}+W_{\rm rest}$, where \[W_{\rm med}=CC\T=\sum_{i\in\cM}W_i, \,\,\,\text{ and  }\,\,\, W_{\rm rest}=\one_V\one_V\T/(100n)+\sum_{i\in\cL}\diag(\one_{S_i})/(100k).\]
We claim that $W_{\rm med}$ has small trace, while $W_{\rm rest}$ has small norm: Indeed,  as each $W_i$ is rank one and $s_i\ge20k^{1/2}$, we have $\tr W_i=\ip{\one}{W_i\one}/s_i\le\ip{\one}{W_i\one}/(20k^{1/2})$, and thus 
$\tr W_{\rm med} \leq \lambda N/(20\sqrt k)$.
And $\|W_{\rm rest}\|\le N/(100n)+ 2k/(100k) \le 3/100$, using that each coordinate lies in at most $2k$ rows. 

Applying Weyl's inequality $\lambda_\ell(A+B)\le\lambda_\ell(A)+\|B\|$, with $A=W_{\rm med}$ and $B=W_{\rm rest}$, it follows that  $W_{\rm med}$ must also have $r$ eigenvalues at least $\lambda/10-3/100\ge7\lambda/100$, using that  $\lambda \ge 1$. As $W_{\rm med}\succeq0$, this implies
 \[r \cdot 7\lambda/100\le \lambda N/(20k^{1/2}),\]
and thus $r\le 5N/(7k^{1/2})\le  N/100$
for $k$ sufficiently large.

\medskip
{\bf Proof of (d).} The codimension of $H_2$ is at most the number of eigenvalues of $M\succeq0$ exceeding $9/11$, which is at most $\frac{11}9\tr M=\frac{11}9\sum_{j\in J}U_{jj}/D_{jj}$. The gradient $\nabla\sigma_i$ is supported on $S_i$, and has entries $(\nabla\sigma_i)_j=-a_{ij}/\sqrt{s_i}+2x_j/s_i$ for $j\in S_i$. So $|a_{ij}|,|x_j|\le1$ and $s_i\ge20k^{1/2}$ give $(\nabla\sigma_i)_j^2\le c_k/s_i$, with $c_k=\bigl(1+O(k^{-1/4})\bigr)$ which approaches $1$ for large $k$. As $D_{jj}=\sum_{i:\,j\in S_i}2\,v\T W_iv/s_i$, we have
\[
 U_{jj}=\sum_{i:\,j\in S_i}(v\T W_iv)\,(\nabla\sigma_i)_j^2\ \le\ (c_k/2)\,D_{jj}\qquad\text{for every } j .
\] So the codimension is at most $(11/18) c_kN\le 5N/8$ for $k$ sufficiently large. 
\end{proof}
\section{The General Koml\'os Problem}\label{sec:komlos}
For a general matrix $A$ with column norm at most $1$, essentially the same argument works with small changes, which we list here. We refer to Guo, Fang, and Lu~\cite{GuoFangLu26} for related computations. 

First, 
the vector $p_i=(a_{ij}^2)_{j\in S_i}$ plays the role of $\one_{S_i}$, and $q_i=\sum_{j\in S_i}a_{ij}^2$ plays the role of the size $s_i$, and $E_i=\sum_{j\in S_i}a_{ij}^2(1-x_j^2)$.
A row is large if $q_i>K^2$, where $K=16$. The slack is $\sigma_i=(\beta-d_i)/\sqrt{q_i}-E_i/(2q_i)-c$, 
where $\beta=2K^3$ and $c=12K$. For $W$, medium rows contribute $W_i\propto e^{-\sigma_i/(2K)}p_ip_i\T/q_i^2$, and large rows contribute $\diag(p_i)/K$.

Second, instead of discarding small rows, we do the following: whenever $a_{ij}^2>q_i/K$ for some $j\in S_i$, row $i$ stops tracking $x_j$. That is, $j$ is removed from $S_i$, and in $d_i$ the coordinate $x_j$ is treated as frozen at its current value, even though it keeps moving. This costs only $O(1)$ in total per row.
As any entry in a row is not too large, this acts like the lower bound $s_i\ge20k^{1/2}$ in the Beck-Fiala analysis.


\bibliographystyle{abbrv}
\bibliography{ref}

\end{document}